\documentclass[journal,11pt]{IEEEtran}
\onecolumn

\usepackage{lipsum}
\usepackage{multirow}
\usepackage[flushleft]{threeparttable}
\usepackage{tablefootnote}

\usepackage[boxed,vlined,ruled]{algorithm2e}
\usepackage[utf8]{inputenc} 
\usepackage[T1]{fontenc}
\usepackage{url}
\usepackage{ifthen}
\usepackage[usenames,dvipsnames,table]{xcolor}

\usepackage[scaled=.80]{beramono}

\usepackage{array,graphicx}
\usepackage{amsfonts,cite,euscript}
\usepackage[cmex10]{amsmath}
\usepackage[amsmath]{empheq}
\usepackage{cases,balance,cite, lineno}
\usepackage{booktabs,mathrsfs,bbm,amssymb,soul,lipsum}
\usepackage{enumitem,amsthm}
\usepackage{dsfont}

\graphicspath{{./Figures/}}

\usepackage{hyperref}
\hypersetup{
    colorlinks=true,
    linkcolor=black,
    filecolor=black,      
    urlcolor=black,
    citecolor = blue,
}

\theoremstyle{plain}
\newtheorem*{theorem*}{Theorem}
\newtheorem{theorem}{Theorem}

\newtheorem*{eg}{Example}

\def\wo{\omega_{\sf{max}}}

\def\Z					{\mathbb Z}
\def\R					{\mathbb R}

\def\T					{T}
\def\iZ					{\in \mathbb Z}
\def\iR					{\in \mathbb R}

\def\e					{{e}}
\def\ind					{\mathbbmtt{1}}
\def\DE					{\stackrel{\rm{def}}{=}}
\def\barwo					{\bar{\omega}_0}

\def\TUS					{T_{\mathsf{US}}}

\def\TFD					{T_{\mathsf{FD}}}

\def\RUS					{\widetilde\gamma_{\mathsf{US}}}
\def\RUSopt				{\widetilde\gamma_{\mathsf{USopt}}}
\def\RFD					{\widetilde\gamma_{\mathsf{FD}}}

\def\lopt					{\lambda_\mathsf{opt}}

\def\l						{\left(}
\def\r						{\right)}

\def\ind					{\mathds{1}}

\def\etal					{\emph{et al}.~}
\def\eg					{\emph{e.g.~}}
\def\ie					{\emph{i.e.}}

\def\trmat					{\boldsymbol{\EuScript{T}}_{M} \rob{\widehat{\bar{\mathbf{r}}}}}
\def\tpmat					{\boldsymbol{\EuScript{T}}_{M} \rob{\mat{z}} }
\def\tymat					{\boldsymbol{\EuScript{T}}_{M} \rob{\widehat{\bar{\mathbf{y}}}}}
\def\txmat					{\boldsymbol{\EuScript{T}}_{M} \rob{\mat{x}} }

\def\setin					{\rob{ \eset{K-1}} + \rob{K-1}\mathbb{Z}}
\def\setout					{\rob{ \mset{K-1} \setminus \eset{K-1}} + \rob{K-1}\mathbb{Z}}

\newcommand\mset[1]		{\mathbb{I}_{#1}}
\newcommand\eset[1]		{\mathbb{E}_{P,#1}}
\newcommand\mcal[1]		{\mathcal{#1}}
\newcommand\rob[1]			{\l #1 \r}
\newcommand\fig[1]			{Fig.~\ref{#1}}

\newcommand{\sqb}[1]		{\left[ #1 \right]}

\newcommand{\ft}[1]			{\left[\kern-0.15em\left[#1\right]\kern-0.15em\right]}
\newcommand{\fe}[1]		{\left[\kern-0.30em\left[#1\right]\kern-0.30em\right]}
\newcommand{\flr}[1]		{\left\lfloor #1 \right\rfloor}

\newcommand{\MO}[1]		{\mathscr{M}_\lambda ({#1} )}
\newcommand{\MONI}[1]		{{\widetilde{\mathscr{M}}}_\lambda ({#1} )}

\newcommand{\VO}[1]		{\varepsilon_{#1}}
\newcommand{\RO}[1]		{\mathscr{R}_{#1}}

\newcommand{\dft}[2]		{\hat{\bar{#1}}\sqb{#2}}
\newcommand{\mdft}[1]		{\hat{\bar{\mathbf{#1}}}}
\newcommand{\dftn}[2]		{\hat{{#1}}_{#2}}
\newcommand{\poly}[1]		{\mathsf{P}_M \rob{#1}}

\newcommand{\BL}[1]		{#1 \in \mathcal{B}_{\Omega}}

\newcommand{\EQc}[1]		{\stackrel{(\ref{#1})}{=}}
\newcommand{\EqRA}[1]		{\stackrel{(\ref{#1})}{\Longrightarrow}}

\newcommand{\vpp}[1]		{{#1}~\mathrm{V}_{\rm{pp}}}

\newcommand{\mat}[1]		{\mathbf{#1}}
\newcommand{\DR}[1]		{\rho_{#1}}

\newcommand{\mse}[1]		{\EuScript{E}{\rob{\gamma,#1}}}

\newcommand{\ep}[1]		{\times 10^{#1}}

\renewcommand\bar\underline
\renewcommand\hat\widehat
\renewcommand\geq\geqslant
\renewcommand\leq\leqslant

\makeatletter
\def\moverlay{\mathpalette\mov@rlay}
\def\mov@rlay#1#2{\leavevmode\vtop{%
   \baselineskip\z@skip \lineskiplimit-\maxdimen
   \ialign{\hfil$\m@th#1##$\hfil\cr#2\crcr}}}
\newcommand{\charfusion}[3][\mathord]{
    #1{\ifx#1\mathop\vphantom{#2}\fi
        \mathpalette\mov@rlay{#2\cr#3}
      }
    \ifx#1\mathop\expandafter\displaylimits\fi}
\makeatother

\newcommand{\cupdot}{\charfusion[\mathbin]{\cup}{\cdot}}

\ifCLASSINFOpdf
\else
\fi
\begin{document}

\title{Unlimited Sampling from Theory to Practice:\\ Fourier-Prony Recovery and Prototype ADC}

\author{Ayush~Bhandari, Felix~Krahmer and Thomas Poskitt

\thanks{A.~Bhandari's work is supported by the UK Research and Innovation council's \emph{Future Leaders Fellowship} program ``Sensing Beyond Barriers'' (MRC Fellowship award no.~MR/S034897/1) and the \emph{European Partners Fund}. F.~Krahmer acknowledges support by the German Science Foundation (DFG) in the context of the collaborative research center TR 109.}
\thanks{A.~Bhandari and T.~Poskitt are with the Dept. of Electrical and Electronic Engineering, Imperial College London, South Kensington, London SW7 2AZ, UK. (Emails: \texttt{ayush@alum.mit.edu} and \texttt{thomas.poskitt15@imperial.ac.uk})}
\thanks{F.~Krahmer is with the Dept. of Mathematics, TU Munich, Boltzmannstra{\ss}e 3, 85748 Garching, Germany. (Email: \texttt{felix.krahmer@tum.de})}

\thanks{Manuscript Submitted: 20XX.}}

\markboth{\sf{Manuscript}}%
{AB \MakeLowercase{\textit{et al.}}: Unlimited Sampling from Theory to Practice}

\maketitle

\begin{abstract}  
Following the Unlimited Sampling strategy to alleviate the omnipresent dynamic range barrier, we study the problem of recovering a bandlimited signal from point-wise modulo samples, aiming to connect theoretical guarantees with hardware implementation considerations. Our starting point is a class of non-idealities that we observe in prototyping an unlimited sampling based analog-to-digital converter. To address these non-idealities, we provide a new Fourier domain recovery algorithm. Our approach is validated both in theory and via extensive experiments on our prototype analog-to-digital converter, providing the first demonstration of unlimited sampling for data arising from real hardware, both for the current and previous approaches. Advantages of our algorithm include that it is agnostic to the modulo threshold and it can handle arbitrary folding times. We expect that the end-to-end realization studied in this paper will pave the path for exploring the unlimited sampling methodology in a number of real world applications. 
\end{abstract}

\tableofcontents

\IEEEpeerreviewmaketitle

\newpage

\section*{Frequently Used Symbols}

\begin{table}[!h]
\normalsize
\centering
\begin{tabular}{p{0.1\columnwidth}p{0.8\columnwidth}}
Symbol & Definition \\ 
$\lambda$       				& Analog-to-digital converter (ADC) threshold.   	     			\\
$\mathscr{M}_{\lambda}$       	& Centered modulo non-linearity.			        				\\
$\widetilde{\mathscr{M}_{\lambda}}$ & Generalized or non-ideal modulo non-linearity.			\\
$t_m$					& Folding instant introduced by $\widetilde{\mathscr{M}_{\lambda}}$. \\
$\mathcal{B}_\Omega$ 		& Space of $\Omega$-bandlimited functions.					\\
$g\rob{t}$       				& Continuous-time, $\Omega$-bandlimited function.				\\
$\gamma\sqb{k}$    			& Point-wise samples of a bandlimited function.            			\\
$y\sqb{k}$    				& Modulo samples of a bandlimited function.            				\\
$\VO{g}\rob{t}$				& Simple function taking values on a $2\lambda$-grid. 			\\
$\RO{g}\rob{t}$				& Simple function taking values on a general grid.		 			\\
$\Delta^N$       				& Finite-difference operator of order $N$.						\\
$\bar{\gamma} \sqb{k}$		& First order finite-difference of $\gamma\sqb{k}$.				\\
$\widehat{g}_p$ 		       	& Fourier series coefficient of function $g\rob{t}$.		          	\\
$\mat{V}$				       	& Discrete Fourier Transform (DFT) matrix.			          	\\
 $\dft{y}{n}$				& Sampled or discrete Fourier transform (DFT).					\\
$\mset{K}$				& Set of $K$ contiguous integers from $0$ to $K-1$.			 	\\
\end{tabular}%
\end{table}

\newpage

\linespread{1}
\section{Introduction}

\IEEEPARstart{I}{n} the recent line of work \cite{Bhandari:2017b,Bhandari:2020g,Bhandari:2020f}, the authors introduced the {\bf Unlimited Sensing Framework} (USF). The USF allows for the acquisition of signals that are orders of magnitude larger than the dynamic range of the analog-to-digital converter (ADC) used in the sampling process. Suppose that an ADC can measure up to $2\lambda$ volts (peak-to-peak), then any signal with maximum amplitude larger than $\lambda$ would result in clipped or saturated samples, for which the Nyquist-Shannon sampling theory is no longer applicable. In contrast, the USF exploits a co-design of hardware and algorithms to allow for \emph {high-dynamic-range} (HDR) signal recovery beyond the threshold of $\lambda$.

\begin{enumerate}[leftmargin =*, label = $\bullet$]
\item On the hardware side, a continuous-time signal is folded via a modulo non-linearity before it is sampled. In \fig{fig:MTCS}, we show an oscilloscope screenshot of the HDR input and the modulo-folded output, obtained via our hardware prototype. This hardware is later used in Section \ref{sec:exp} of the paper to validate the theory presented in this work. 
\item On the algorithmic side, one needs to solve the ill-posed inverse problem of recovering a signal from folded measurements. The solution approach of \cite{Bhandari:2017b,Bhandari:2020g,Bhandari:2020f}, which capitalizes on certain commutativity properties of the modulo non-linearity, and the associated reconstruction guarantees are reviewed in the next subsection.
\end{enumerate}

\begin{figure}[!b]
\centering
\includegraphics[width = 0.65\columnwidth]{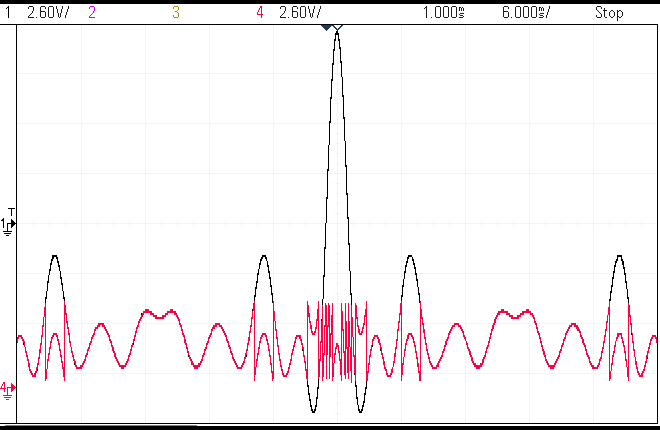}
\caption{Hardware demonstration for \emph{unlimited sampling} \cite{Bhandari:2017b,Bhandari:2020g,Bhandari:2020f}. We show an oscilloscope screen shot plotting a continuous-time function (ground truth, black) and its folded version (pink). The input signal with dynamic range $20 \mathrm{V}$ peak-to-peak ($\approx 10\lambda$) is folded into a $4.025 \rm{V}$ peak-to-peak, signal. We have tested sampling and recovery of signals as large as $24\lambda$. A live {YouTube} demonstration is available at \href{https://youtu.be/prV40WlzHh4}{\texttt{https://youtu.be/prV40WlzHh4}}.}
\label{fig:MTCS}
\end{figure}

\subsection{Overview of Unlimited Sampling and Reconstruction} 
When working with folded signals as in \fig{fig:MTCS}, the following result shows that a bandlimited function can be recovered from a constant factor oversampling of its modulo samples.

\begin{theorem}[Unlimited Sampling Theorem \cite{Bhandari:2017b}]
\label{thm:UST}
Let $f\l t \r$ be a continuous-time function with maximum frequency $\wo$ (rads/s). Then, a sufficient condition for recovery of $f\l t \r$ from its modulo samples (up to an additive constant) taken every $T$ seconds apart is $T\leqslant 1/ \l 2 \wo \e\r$ where $\e$ is Euler's constant.
\end{theorem}
Note that the sampling criterion is independent of $\lambda$ and only depends on the bandwidth of the signal. This is somewhat surprising and indeed the reconstruction becomes less stable with respect to noise for large amplitudes \cite{Bhandari:2020f}. Also, the sampling rate required for local reconstruction from a finite set of modulo samples grows in proportion to the dynamic range of the signal, relative to $\lambda$, see for example the \emph{local reconstruction theorem} discussed in \cite{Bhandari:2018b,Bhandari:2018c}.

\medskip

{{\noindent {\bf How does the Unlimited Sampling algorithm \cite{Bhandari:2020f} work?}}} Let us formally define the centered modulo operation using the mapping 
\begin{equation}
\label{map}
\mathscr{M}_{\lambda}:f \mapsto 2\lambda \left( {\fe{ {\frac{f}{{2\lambda }} + \frac{1}{2} } } - \frac{1}{2} } \right), 
\quad \ft{f} \DE f - \flr{f} 
\end{equation}
where $\ft{f}$ and $\flr{f}$ define the fractional part and floor function, respectively. The recovery procedure for ``inverting'' the $\MO{\cdot}$ operator relies on two steps, (i) isolating the higher order finite differences followed by, (ii) stably inverting the difference operator. The modulo decomposition property (cf.~ Fig.~4 in \cite{Bhandari:2020f}) allows us to write, 
\begin{equation}
\label{MDP}
g\rob{t} = \MO{g\rob{t}} + \VO{g}\rob{t},\qquad \VO{g} \in 2\lambda\mathbb{Z}
\end{equation}
where $\VO{g}$ is a simple function. Let us fix $\BL{g}$ where $\mathcal{B}_\Omega$ denotes the space of $\Omega$-bandlimited functions. We denote the conventional and modulo samples by, $\gamma\sqb{k} = g\rob{kT}$ and $y\sqb{k} =\MO{g\rob{kT}}$, respectively.
Let ${\Delta ^N}y = {\Delta ^{N - 1}}\left( {\Delta y} \right)$ denote the $N^{\rm{th}}$ difference operator with $\rob{\Delta y}\sqb{k} = y\sqb{k+1}-y\sqb{k}$. Since ${\Delta ^N}\VO{g} \in 2\lambda\Z $, it follows that, 
\begin{equation}
\label{eq:HOME}
\MO{\Delta^N\VO{g}\rob{kT}} = 0 \ \EqRA{MDP} \  \MO{\Delta^N \gamma} = \MO{\Delta^N y}.
\end{equation}
Oversampling $\BL{g}$ results in highly correlated samples and hence, $\rob{\Delta^N\gamma}\sqb{k}$ shrinks as the sampling step $T$ decreases. Quantitatively, the shrinking effect is explained by the bound, $\|\Delta^N \gamma\|_\infty \leq {\left( {\T \Omega \e} \right)^N} \| g\|_\infty$ in \cite{Bhandari:2020f} where $\| \cdot\|_\infty$ denotes the max-norm. 
For a suitable $N$ (cf.~\cite{Bhandari:2020f}), namely,
\begin{equation}
\label{eq:NO}
N^\star \geq \left\lceil {\frac{{\log \lambda  - \log \beta_g}}{{\log \left( {T\Omega \e} \right)}}} \right\rceil, \quad \beta_g \in 2\lambda\Z \mbox{ and } \beta_g\geq \| g\|_\infty,
\end{equation}
choosing $T\leq 1/\Omega\e$ ensures that $\|\Delta^{N^\star} \gamma\|_\infty \leq \lambda$. Modular arithmetic shows that\footnote{See Proposition 2 in \cite{Bhandari:2020f}.} for any sequence $s\sqb{k}$, it holds that 
\begin{equation}
\label{eq:modcom}
{\mathscr{M}_{\lambda}(\Delta^N s) =\mathscr{M}_{\lambda}(\Delta^N({\mathscr{M}_{\lambda}(s)})}.
\end{equation}
By choosing $N = N^\star$, we have, 
\begin{equation}
\label{eq:US}
T = \TUS\leq 1/\Omega\e \ \EqRA{eq:modcom} \ \Delta^{N^\star}  \gamma = \MO{\Delta^{N^\star} y}.
\end{equation}
The unlimited sampling based recovery algorithm \cite{Bhandari:2017b} recovers $\gamma$ from $\Delta^{N^\star}\gamma$ and is also stable with respect to quantization noise \cite{Bhandari:2020f}. The reconstruction works by estimating $\Delta^{\rob{N^\star-n}} \VO{\gamma}$ for $n=\sqb{0,N^\star}$ thus yielding $\VO{\gamma}$, and finally, $\gamma = y + \VO{\gamma}$, is the recovered signal. The approach in \cite{Bhandari:2017b,Bhandari:2020f}, 
\begin{enumerate}[leftmargin =15pt, label = {\upshape(\roman*)}]
  \item is inherently stable because we are able to exploit the restriction on the range of $\VO{\gamma}$, that is,  $\VO{\gamma} \in 2\lambda\mathbb{Z}$.
  \item exploits properties of bandlimited functions\footnote{
  In praticular, we use Bern\v{s}te\u{\i}n's inequality. This allows for recovery of the unknown polynomial in the kernel of $\Delta^N$, up to a constant.} to estimate the ``unknown constant of integration'' for inversion of $\Delta^N$.  
\end{enumerate}
It is precisely this synergistic interplay between (i) and (ii) that allows our approach to treat orders $N>1$,  distinguishing it with seemingly similar methods\footnote{The key difficulty is that each time the operator $\Delta$ is inverted, an unknown constant in the kernel of $\Delta$ has to be estimated. 
Also for $N=1$, our method is considerably more stable as these approaches do not capitalize on the stabilizing effect described in (i).
}  
\cite{Rieger:2009} analogous to Itoh's method 
for phase unwrapping \cite{Itoh:1982}, that are restricted to low orders. For further details and comparisons, see \cite{Bhandari:2020f}. Recovery with higher orders has clear advantages in the context of applications such as HDR imaging \cite{Bhandari:2020}, tomography \cite{Bhandari:2020a,Beckmann:2020} and sensor array processing \cite{FernandezMenduina:2020,FernandezMenduina:2020a}.

\subsection{Related Work} 
Following our work \cite{Bhandari:2017b}, the problem of sampling and reconstruction of bandlimited and smooth functions, from modulo samples, has been studied in various contexts. Ordentlich \etal \cite{Ordentlich:2018} studied recovery from quantized modulo samples, using side-information, from a rate-distortion perspective. In parallel papers, \cite{Bhandari:2019a} and \cite{Romanov:2019} proved that bandlimted functions are uniquely characterized by modulo samples, when sampling rate is above the Nyquist rate. In \cite{Romanov:2019}, the authors also present a constructive approach, provided that a subset of unfolded samples is known. Unlimited sampling has also been studied in the context of finite-length signals \cite{Bhandari:2018b,Bhandari:2018c}, random measurements of sparse signals \cite{Musa:2018}, one-bit \cite{Graf:2019} and multi-channel sampling \cite{Gan:2020} as well as wavelet based reconstruction approaches \cite{Rudresh:2018}. Recovery guarantees for denoising of modulo samples with bounded and Gaussian noise models were discussed in \cite{Cucuringu:2018}. 

Recently, we have have developed USF based recovery methods that are tailored to larger classes of signal spaces and inverse problems. For instance, recovery of multi-dimensional functions on arbitrary lattices was considered in \cite{Bouis:2020}. Functions that belong to spline spaces (\eg images) were studied in \cite{Bhandari:2020}. The modulo Radon transform was introduced in \cite{Bhandari:2020a} and its application to HDR tomography was presented in \cite{Beckmann:2020}. Computational sensor array signal processing based on USF was presented in \cite{FernandezMenduina:2020,FernandezMenduina:2020a}.    

\begin{figure*}[!t]
\centering
\includegraphics[width = 1\textwidth]{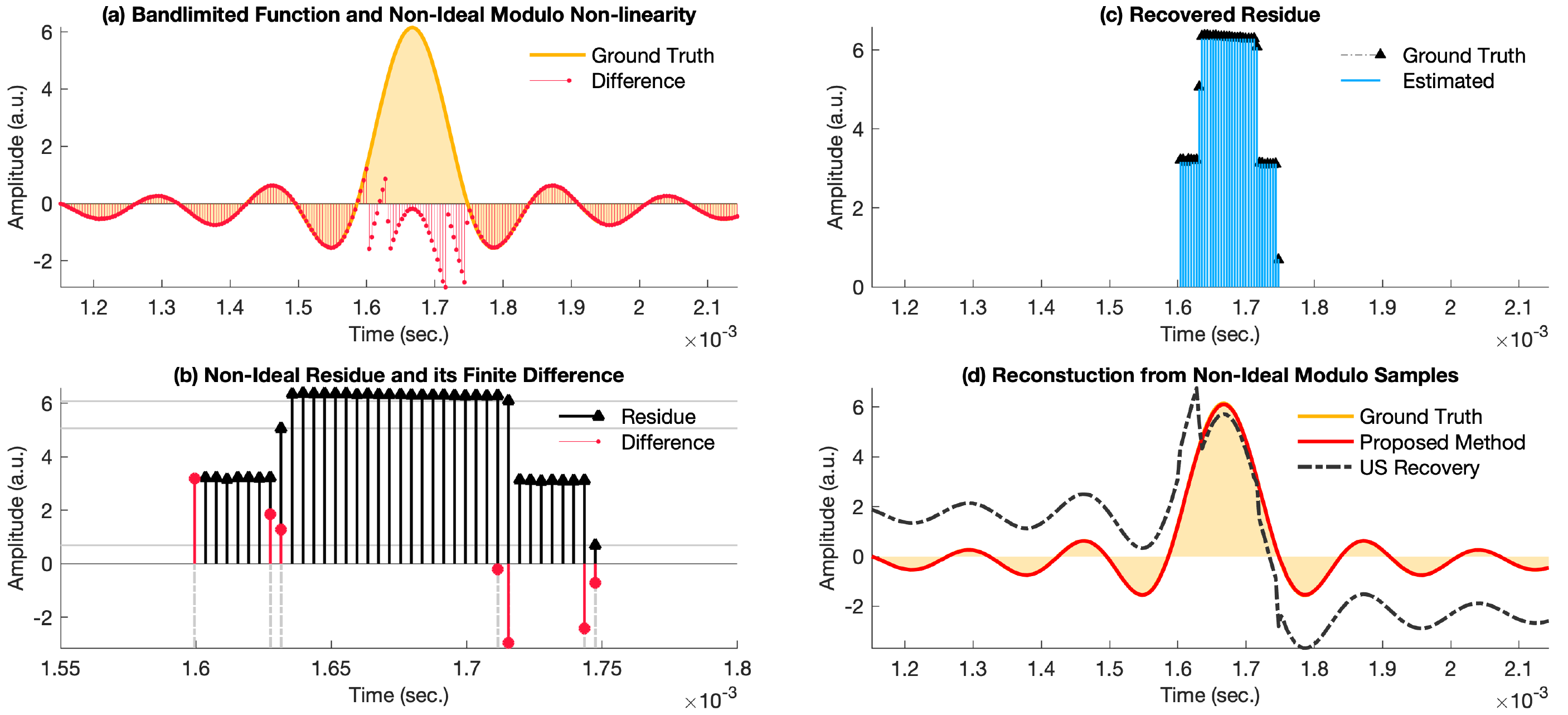}
\caption{Example of reconstruction based on data acquired using our prototype US-ADC. Ground truth is shown in shaded yellow. (a) Non-ideal modulo mapping. (b) Unlike in the perfect modulo case, the residue function, or $g-\MO{g}$, no longer lies on a grid of $2\lambda\Z$, specially at the locations marked by gridlines on the y-axis. The first order difference encodes both non-ideal folds and inaccuracies in terms of a sparse representation. (c) By exploiting the sparse structure of the residue, our algorithm that is agnostic to $\lambda$ is able to recover the residue precisely. (d) Since $g-\MO{g}\not\in2\lambda\Z$, \eqref{eq:HOME} does not apply and unlimited sampling based recovery \cite{Bhandari:2020f} yields an erroneous reconstruction. In contrast, our Fourier domain approach based reconstruction is agnostic to $\lambda$ and results in a near perfect reconstruction.}
\label{fig:demo1}
\end{figure*}

\subsection{Motivation and Contributions} 
The work presented in this paper is pivoted around the practical aspects of unlimited sampling and the insights developed from building a prototype USF based ADC. 

\begin{enumerate}[leftmargin=*]
\itemsep 5pt
  \item {\bf Non-Ideal Folding.} When implementing the modulo circuit in hardware {(see Section~\ref{sec:exp} for more details)}, we observed that it occasionally exhibits non-ideal foldings. An example output of the circuit with such artifacts is shown \fig{fig:demo1}(a). More precisely, in \fig{fig:demo1}(a) some of the folding times are delayed; other types of non-ideal foldings that we have observed includes spurious jumps and inaccuracies in the folding threshold $\lambda$. In all these cases, the residue $\VO{g}$ in \eqref{MDP} {is still piecewise constant}, see for example \fig{fig:demo1}(b), but no longer satisfies $\VO{g} \in 2\lambda\mathbb{Z}$. The consequence is that $\MO{\Delta^N\VO{g}\rob{kT}} \not = 0$ in \eqref{eq:HOME} and the reconstruction via \eqref{eq:US} is erroneous, see~\fig{fig:demo1}(c)). We attribute the artifacts to electronic limitations such as the maximum rate at which an input signal can be folded. As this rate depends on the bandwidth of the input signal, $\wo$, an exact implementation must be carefully calibrated to match $\wo$ in Theorem \ref{thm:UST}. In this paper we show that this calibration is not required and rather, such limitations can be circumvented algorithmically.
  
  \item {\bf Lower Sampling Rates.} While $T\Omega \e \leq 1$ guarantees that \eqref{eq:US} holds for some $N$, in practice, for stability reasons, it is desirable that we work with smaller values of the finite-difference order $N$. To satisfy  \eqref{eq:NO}, one then needs higher oversampling. Thus a natural question is whether alternative recovery approaches can allow for recovery with both moderate oversampling and low values of $N$. 
\end{enumerate} 

\medskip
\noindent{\bf Contributions:}
{The main contribution of this work is to provide the first validation the unlimited sampling approach, thus taking our theoretical ideas all the way to practice. To do so, we go beyond the conventional literature in fundamental ways. On the theory front, we propose a novel, Fourier domain, recovery algorithm that can handle non-idealities and uncertainties introduced by the hardware, while operating at lower sampling rates. On the hardware end, we develop a custom-designed, USF based prototype ADC, the \emph{US-ADC}. This is the key to enabling real experiments. Extensive hardware experiments based on the US-ADC, corroborate the effectivity of our new recovery approach. The upshot of our \emph{end-to-end} sensing pipeline is that we can recover signals as large as $24\times$ the ADC threshold $\rob{\lambda}$. Furthermore, our work also validates the first approaches presented in \cite{Bhandari:2020f}.} 

Concretely, the advantages of our recovery method include that it
\begin{enumerate}[leftmargin=*,label = $\arabic*)$]
  \item is agnostic to $\lambda$ and hence, can combat any non-idealities. 
  \item requires computation of $\Delta^1$ only; this is specially beneficial in the case of errors. {The approach in \cite{Bhandari:2020f} requires computation of $\Delta^N$ which can be sensitive to hardware artifacts.}
\end{enumerate}
The precise signal model that we are working with consists of periodic, bandlimited signals, \ie, trigonometric polynomials of finite degree. Note that this model is more restrictive than the infinite dimensional model considered in previous works; this restriction, however, reflects the practical limitation that one typically samples signals on a finite interval rather than the full real line.

\medskip
\noindent {\bf Notation.} The sets of real, integer, and complex-valued numbers are denoted by $\mathbb{R}$, $\mathbb{Z}$ and $\mathbb{C}$, respectively. We use $\mset{K} = \{0,\ldots,K-1\}, K\iZ^+$ to denote the set of $K$ contiguous integers while its continuous counterpart is denoted by $\ind_{\mcal{X}}\rob{t}, t\iR$, the indicator function on the domain $\mcal{X}$. For a $\tau$-periodic function $h$, we consider the {renormalized} Fourier series coefficients as given by $\dftn{h}{m} =\int_{0}^{\tau} h\rob{t}e^{-\jmath m\omega_0 t} dt$ where $\omega_0 = 2\pi/\tau$ is the fundamental harmonic.  Vectors and matrices are written in bold fonts. The mean squared error or MSE between vectors $\mat{x}$ and $\mat{y}$ of length $K$ is defined by 
\begin{equation}
\label{eq:MSE}
\EuScript{E}\rob{\mat{x},\mat{y}} \DE \frac{1}{K}\sum\limits_{k=0}^{K-1} \vert x\sqb{k} - y\sqb{k} \vert^2.
\end{equation}

\section{Fourier Domain Reconstruction Approach}

\noindent{\bf Signal Model.} In our work, we consider $\BL{g}$ such that $g\left( t \right) = g\left( {t + \tau } \right),\forall t \in \mathbb{R}$. Such signals can be written as,
\begin{equation}
\label{gfs}
g\left( t \right) = \sum\limits_{\left| p \right| \leqslant P} {\hat{g}_p}
{e^{\jmath p {\omega _0}t}} ,\quad {\omega _0} = \frac{{2\pi }}{\tau }, \quad P = 
\left\lceil {\frac{\Omega }{{{\omega _0}}}} \right\rceil 
\end{equation}
where $\widehat{g}_p$ denotes the Fourier series coefficient and $\hat{g}_{-p} = \widehat {g}^*_p$ (Hermitian symmetry) with ${\sum\limits_p {\left| {\hat g_p} \right|} ^2} < \infty$. Sampling $g\rob{t}$ with sampling rate $T$ results in $K$ samples $\gamma\sqb{k}$ on the interval $\left[ {0,\tau } \right)$. To solve for $\hat{g}_p$ in \eqref{gfs}, one requires $K\geq2P+1$. Wen working with oversampled representation, we write $\dftn{g}{p}$ in the Fourier domain as, 
\begin{equation}
\label{eq:osdft}
\dftn{g}{p} =
\frac{1}{\tau}
\begin{cases}
\int_{0}^{\tau} g\rob{t} e^{-\jmath \omega_0 p t} , \quad & p \in \eset{K} \\
0 ,\quad &p \in \mset{K} \setminus \eset{K}
\end{cases},
\end{equation}
where the set $\eset{K}$ is given by, 
\begin{equation}
\label{eset}
\eset{K} = \left[ {0,P} \right] \cup \left[ {K - P,K - 1} \right], \ \ \hfill |\eset{K}| = 2P+1.
\end{equation}
With $KT = \tau$, the critical sampling rate for a $P = \left\lceil {\Omega /{\omega _0}} \right\rceil$ bandlimited function is $T\leq\tau/\rob{2P+1}$.

\medskip

\noindent {\bf Generalized Modulo Non-linearity.} Our \emph{non-ideal} modulo non-linearity, $\MONI{\cdot}$, corresponds to subtracting a piecewise constant function, with finite discontinuities, from $\BL{g}$. In formulas, $\MONI{\cdot}$ admits a representation of the form,
\begin{equation}
\label{eq:monis}
g\rob{t} = \MONI{g\rob{t}} + \RO{g}\rob{t}, \quad t \in [0,\tau)
\end{equation}
where the (generalized) residue function is of the form 
\begin{equation}
\label{eq:res}
\RO{g}\rob{t} = \sum\limits_{m \in \mcal{M}} {c\sqb{m}{\ind_{{\mathcal{D}_m}}}\left( t \right)}, 
\ \ 
\cupdot_m \ind_{{\mathcal{D}_m}} = \R,
\ \ 
c\sqb{m}\iR.
\end{equation}
Unlike in \cite{Bhandari:2020f}, we make no assumptions on the coefficients $c\sqb{m}$ related to the residue function $\RO{g}\rob{t}$. Note that $\MONI{\cdot}$ includes the exact modulo operation as a special case and hence our algorithm is backwards compatible with the sampling model in \cite{Bhandari:2020f}.

\subsection{Recovery Approach}

\begin{figure*}[!t]
\centering
\includegraphics[width = 1\textwidth]{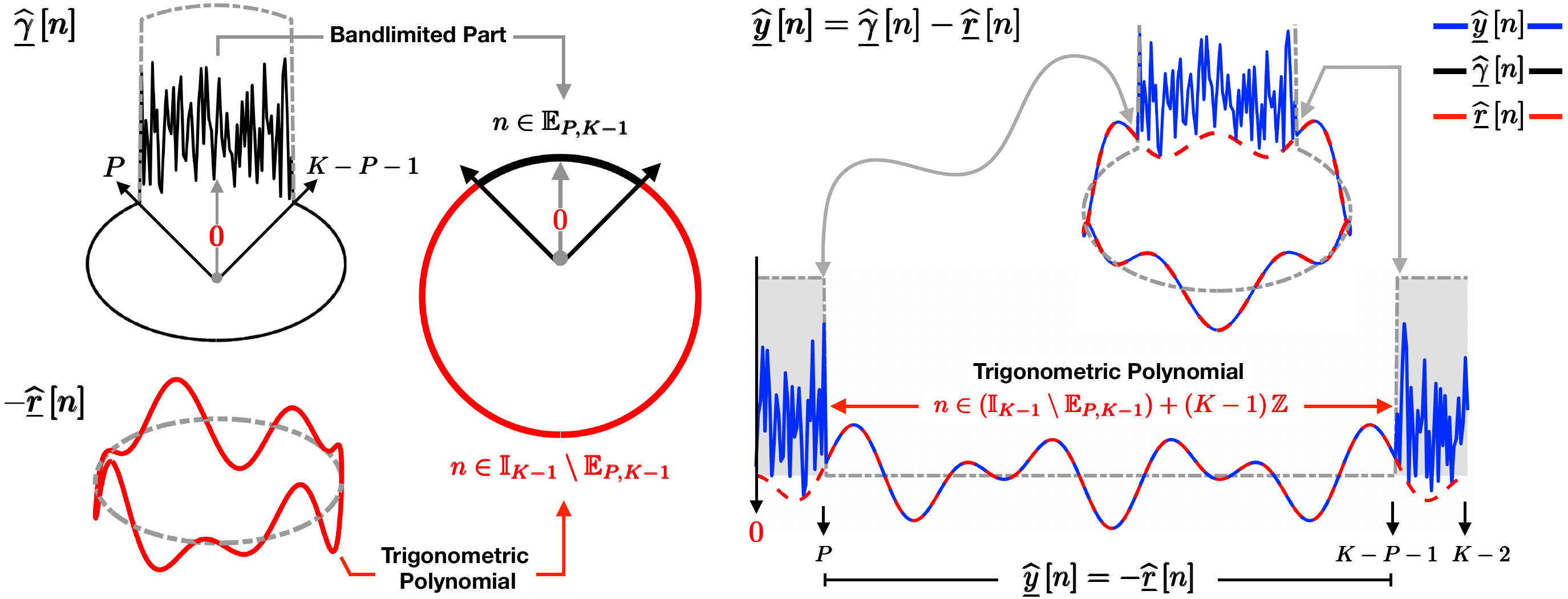}
\caption{Fourier domain partitioning of the bandlimited signal and the non-ideal folding instants given by \eqref{eq:ydft}.}
\label{fig:dft}
\end{figure*}

Letting $\delta$ denote the Dirac distribution and
\begin{equation}
\label{eq:barvec}
\bar y \sqb{k} \DE \Delta y \sqb{k}, \ \ 
\bar{\gamma} \sqb{k}  \DE \Delta \gamma \sqb{k},  \ \  \mbox{and} \ \ 
\bar{r} \DE \Delta\RO{\gamma}\rob{kT},
\end{equation}
respectively, we note that \eqref{eq:monis} implies 
\begin{align}
k\in \mset{K},\ \ \bar y \sqb{k} 	& =  \bar\gamma \sqb{k} - \bar r \sqb{k}  \notag \\
			& = \bar\gamma \sqb{k} - \sum\limits_{m \in \mcal{M}} {c\sqb{m}\delta \left( {kT - {t_m}} \right)}
\label{eq:spike}			
\end{align}
where $t_m \in \rob{T\Z}\cap \left[0,\tau\right)$ are the \emph{unknown} folding instants. The size of the set $\mcal{M}$ depends on the dynamic range of the signal relative to the threshold $\lambda$. Since $c\sqb{m} \not\in 2\lambda\Z$, the contribution due to $\RO{g}$ cannot be removed by non-linear filtering of the amplitudes using a modulo non-linearity as in \eqref{eq:HOME}, which is a key ingredient of the method in \cite{Bhandari:2020f}. A classical alternative would be to filter out the ``impulsive'' component $\bar r \sqb{k} = \rob{\Delta\RO{\gamma}}\sqb{k}$ using the median filter. However, this yields an inexact solution due to the nature of the median filter {and is not considered in this paper}.

Instead, given $K$ modulo samples $y\sqb{k}, k\in\mset{K}$, we exploit the fact that \eqref{eq:spike} can be partitioned in the Fourier domain as, 
\begin{align}
\label{eq:ydft}
\dft{y}{n} & = 
\begin{cases}
\dft{\gamma}{n} - \dft{r}{n}
& n\in \setin \\ 
- \dft{r}{n}
& n \in \setout
\end{cases}
\end{align}
where $\dft{y}{n}$, the sampled or discrete Fourier transform (DFT) of $y\sqb{k}$, is given by, 
\begin{equation}
\label{eq:dft}
\dft{y}{n} \DE \sum\limits_{k \in \mset{K-1}} {\bar{y}\sqb{k}{e^{ - \jmath \barwo nk}}}, \quad \barwo = \frac{2\pi}{K-1}.
\end{equation}
The Fourier domain partition in terms of the DFT coefficients distributed over the two sets, 
\begin{itemize}
  \item $n \in \setin $ (bandlimited part) and, 
  \item $ n \in \setout$ (modulo part)
\end{itemize}
is schematically explained in \fig{fig:dft}. In vector-matrix notation, we can re-write \eqref{eq:dft} as, 
\begin{equation}
\label{eq:dftvec}
\mdft{y} \DE \mat{V} \mat{\bar{y}}, \qquad 
{\left[ \mat{V} \right]_{n,k}} = {e^{ - \jmath \left( {\frac{{2\pi }}{{K - 1}}} \right)nk}}, \qquad 
k\in\mset{K-1}
\end{equation}
where $\mat{V} \in \mathbb{C}^{\rob{K-1}\times\rob{K-1}}$ is the DFT matrix\footnote{We remind the reader that starting with $K$ modulo samples, the operation $ \bar y \sqb{k} = \Delta y \sqb{k}$ results in a loss of sample and hence the square matrix $\mat{V}$ has a dimension (or rank) of $K-1$.} and $\mdft{y}$ is the DFT vector corresponding to the first-order finite-difference samples $\mat{\bar{y}}$. In \eqref{eq:ydft} and also as shown in \fig{fig:dft}, we have $\bar{y}\sqb{n} = - \bar{r}\sqb{n}$
, $\forall n \in { \mset{K-1} \setminus \eset{K-1}}$ and this simplifies to a sum of complex exponentials,
\begin{align}
\dft{r}{n}  &= \sum\limits_{k\in\mset{K-1}} {\sum\limits_{m\in\mcal{M}} {c\sqb{m}\delta \left( {kT - {t_m}} \right){e^{ - \jmath \barwo nk}}} } \notag  \\
\label{eq:SOCE}
& = \ \  \sum\limits_{m\in \mcal{M}} {c\sqb{m}{e^{ - \jmath \frac{\barwo n}{T}{t_m}}}}
\end{align}
with $M = |\mcal{M}|$ terms. {Estimating the unknown parameters in \eqref{eq:SOCE} boils down to the spectral estimation problem \cite{Kay:1988:Book}.} In the area of \emph{error correction coding} and \emph{impulse} cancellation, Wolf \cite{Wolf:1983} first observed that the non-bandlimited spikes (or $\bar{r}\sqb{k}$) being a parametric function, can be curve-fitted (\eg using Prony's method) on the ``out-of-band'' interval, namely, ${ \mset{K-1} \setminus \eset{K-1}}$ in the Fourier domain. This approach was also used for impulse noise removal from images \cite{Rioul:1996} as well as bandlimited functions, in the context of Fourier \cite{Marziliano:2006} and generalized-Fourier domain \cite{Bhandari:2015a} sampling theory.

In our case, we notice that $\rob{\jmath\omega}\widehat{\RO{g}}\rob{\omega}$ or equivalently $\dft{r}{n}$ in \eqref{eq:SOCE} is the contribution due to non-bandlimited spikes. To see the approach in action, let us define a polynomial $\poly{z}, z\iZ$ of degree $M = |\mcal{M}|$,
\begin{align}
\label{proots}
& \poly{z} = \prod\limits_{m = 0}^{M-1} {\left( {1 - \frac{\xi_m}{z}} \right)}  = \sum\limits_{n = 0}^{M} {{p\sqb{n}}{z^{ - k}}} \\ 
\label{eq:roots}
\mbox{with,} \ \    & \mathrm{roots}\rob{\poly{z}} = \{\xi_m\}_{m=0}^{M-1}, \quad \xi_m = \e^{-\jmath \frac{\barwo}{T} t_m}.
\end{align}
To evaluate the roots of the polynomial $\poly{z}$, one starts with identification of the filter coefficients\footnote{We use $M$ unknowns of $p\sqb{m}$ instead of $M+1$ as given by \eqref{proots} as we normalize $p\sqb{M} = 1$.} $\{p\sqb{m}\}_{m=0}^{M-1}$ using Prony's Method. On any interval of size $2M$, we note that $\rob{p*\hat{\bar{r}}}\sqb{n} =0$, which is because, 
\begin{align}
\rob{p*\hat{\bar{r}}}\sqb{l}& = \sum\limits_{n=0}^{M} {{p\sqb{n}}\dft{r}{l-n}}  \hfill \notag \\
& \EQc{eq:SOCE} \sum\limits_{m=0}^{M-1} {c\sqb{m} 
{\sum\limits_{n=0}^{M} {p\sqb{n}{e^{\jmath \frac{\barwo}{T}n{t_m}}}} } 
{e^{ - \jmath \frac{\barwo}{T}l{t_m}}}} \notag \\
& \EQc{proots} \sum\limits_{m=0}^{M-1} {c\left[ m \right] 
\underbrace{\poly{{e^{-\jmath \frac{\barwo}{T}{t_m}}}} }_{\poly{\xi_m}=0} {e^{ - \jmath \frac{\barwo}{T}l{t_m}}}} = 0.
\label{prony}
\end{align} 
{This leads to a recipe for evaluation of unknown filter $\mat{p} = \sqb{\ p\sqb{0}, \cdots, p\sqb{M} \ }^\top$ with $p\sqb{M} = 1$.} Algorithmically, \eqref{prony} implies that the coefficients of the unknown $M$--tap FIR filter, $\mat{p}$ are in the kernel of a Toeplitz matrix comprising of out-of-band samples of $\dft{y}{n}$. More precisely, for a given vector $\mat{x}$, let us define a Toeplitz matrix\footnote{In \texttt{MATLAB}, given a vector \texttt{x} defined by $2M$ contiguous samples, namely, \texttt{x[-M+L]} $\cdots$ \texttt{x[M-1+L]} where \texttt{L} is any integer-valued translate, the corresponding toeplitz matrix $\txmat$ is conveniently obtained by defining the function handle \texttt{T = @(x,M) toeplitz(x(M+1:end),x(M+1:-1:1))}.} $\txmat \in \mathbf{C}^{M\times\rob{M+1}}$ as, 
\begingroup
\renewcommand{\arraystretch}{1.1}
\begin{equation}
\label{tmatrix}
\txmat \DE 
\left[ {\begin{array}{*{20}{c}}
  {x\left[ 0 \right]}&{x\left[ { - 1} \right]}& \cdots &{x\left[ { - M} \right]} \\ 
  {x\left[ 1 \right]}&{x\left[ 0 \right]}& \cdots &{x\left[ {1 - M} \right]} \\ 
   \vdots & \vdots & \ddots & \vdots  \\ 
  {x\left[ {M - 1} \right]}&{x\left[ {M - 2} \right]}& \cdots &{x\left[ { - 1} \right]} 
\end{array}} \right].
\end{equation}
\endgroup
Then, \eqref{prony} amounts to, 
\begin{equation}
\label{eq:pest}
{\mathbf{p}} \in \ker \rob{\trmat} \  \Longleftrightarrow \ \trmat  \mat{p} = \mat{0}.
\end{equation}
By construction, $\trmat$ requires any $2M$, contiguous samples for identifying $M = |\mcal{M}|$ folds of the modulo non-linearity. As per our definition in \eqref{eq:ydft}, $\dft{r}{n}$ is isolated on the set $ \mset{K-1} \setminus \eset{K-1}$. For the estimation of $\mat{p}$ in \eqref{eq:pest}, the requirement is that the cardinality of the set $\mset{K-1} \setminus \eset{K-1}$ is larger than $2M$. 
From \fig{fig:dft}, we see that $|\mset{K-1} \setminus \eset{K-1}| = K -2P- 2$ and this translates to the condition,  
\[
K -2P- 2 \geqslant 2M \  \Longleftrightarrow \  K \geqslant 2\rob{ P  + M + 1}
\]
With $KT=\tau$, the sampling density criterion, 
\begin{equation}
\label{eq:sampling}
T = \TFD \leq \frac{\tau} { 2\left( P  + M + 1 \right)},  
\quad P = \left\lceil {\frac{\Omega }{{{\omega _0}}}} \right\rceil 
\end{equation}
guarantees the recovery of the $M$ folding instants $\{t_m\}_{m=0}^{M-1}$ introduced by the non-ideal operator $\MONI{\cdot}$. We formalize our recovery guarantee in the following theorem.

\begin{theorem}[Fourier Domain Reconstruction] 
\label{thm:FDUST}
Let $\BL{g}$ be a $\tau$-periodic function. Suppose that we are given $K$ modulo samples of $y\sqb{k} = \MONI{g\rob{kT}}$  folded at most $M$ times. Then a sufficient condition for recovery of $g\rob{t}$ from $y\sqb{k}$ (up to a constant) is that, $T\leq \tau/K$ and $K \geqslant 2\left( {\left\lceil {\frac{{\Omega \tau }}{{2\pi }}} \right\rceil  + M + 1} \right)$.
\end{theorem}

\begin{proof}
Note that, by \eqref{eq:ydft} one has that
\[
\dft{y}{n} = -\dft{r}{n}, \quad \forall  n  \in \setout.
\]
Then, observing that only such $n$ appear on the right hand side of \eqref{tmatrix}, we obtain from \eqref{eq:pest} that ${\mathbf{p}} \in \ker \rob{\tymat}$. Given that $\tymat$ has rank $M$ (this follows from the Vandermonde decomposition of Toeplitz matrices, similar to the Carath\'{e}odory--Fej\'er decomposition, cf.~\cite{Yang:2016}), this implies that $\widetilde{\mat{p}} = \nu \mat{p}$ for some $\nu \in \mathbb{C}$ and $\widetilde{\mat{p}}$ as obtained in step $4\textit{b})$ of Algorithm~\ref{alg:1}. Consequently, the roots of $\widetilde{\xi}_m$ of ${\widetilde{\mathsf{P}}_M \rob{z}}$ agree with the roots ${\xi}_m$ of $\poly{z}$ and one has $\widetilde{t}_m = t_m$ up to index permutation. Once the exponents have been identified, plugging \eqref{eq:SOCE} into \eqref{eq:ydft} yields a linear system which can be solved by least-squares minimization as in Step $4\textit{e})$. This implies that $\widetilde{c}\sqb{m}  = c\sqb{m}$ and hence $\widetilde {\bar r}\sqb{k}$ = $ {\bar r}\sqb{k}$, and therefore $\widetilde{r}\sqb{k}$ = ${r}\sqb{k}$ (up to an additive constant). Together with the modulo samples, these residuals allow for the recovery of the unfolded samples, which completes the proof. 
\end{proof}

{\noindent \bf {Reconstruction without Periodicity.}} Recall that the periodicity assumption made in our paper enables a practical approach for recovery of signals from folded measurements. A theoretical reconstruction guarantee when infinite samples are available, however, can also be obtained for signals that are not periodic via the discrete-time Fourier transform (DTFT). Then the requirement will be to have a sampling density high enough to ensure that in the DTFT, the shift between the aliased copies of the frequency support is at least $\Omega+2M+1$.

\subsection{Summary of Recovery Algorithm}
{Here, we qualitatively summarize the rationale of Algorithm~\ref{alg:1}. When the sampling criterion in \eqref{eq:sampling} is satisfied, our recovery method can be applied to ``unfold'' the non-ideal, modulo samples. Starting with $K$ folded samples of the bandlimited function in \eqref{gfs}, we compute the first-difference $\bar{y} = \Delta y$ (cf.~\eqref{eq:barvec}). Then, we apply the DFT on the vector  $\bar{y}$ consisting of $K-1$ samples, yielding $\mdft{y}$ in \eqref{eq:dftvec}. According to \eqref{eq:ydft}, the DFT coefficients defined for $n \in \setout$ are solely attributed to the \emph{unknown} folding parameters $\{c\sqb{m},t_m\}_{m\in\mcal{M}}$ of $\RO{g}\rob{t}$ in \eqref{eq:res}. Hence, in the Fourier domain, we define $\sqb{\mat{z}}_n = \sqb{\mdft{y}}_n, \forall n\in \mset{K-1} \setminus \eset{K-1}$ which follows the parametric representation in \eqref{eq:SOCE}. Given $M = |\mcal{M}|$, estimation of $\{c\sqb{m},t_m\}$ boils down to the classical \emph{spectral estimation problem} \cite{Kay:1988:Book}. Concretely, this is implemented in step $4)$ of Algorithm~\ref{alg:1} yielding estimates $\{\widetilde c\sqb{m}, \widetilde t_m\}$ which are used to estimate $\widetilde {\bar r}\sqb{k}$ (cf.~step $5)$). To map $\widetilde {\bar r}\sqb{k}\to {\widetilde r}\sqb{k}$, we need to invert the first-difference operator and this is carried out in (cf.~step $6)$. Combining ${\widetilde r}\sqb{k}$ with modulo samples $y\sqb{k}$ in accordance with \eqref{eq:spike} yields bandlimited samples $\RFD\sqb{k}$ (cf.~step $7)$). Low-pass filtering the same results in the continuous-time function $\BL{\widetilde{g}}$.}

\begin{algorithm}[!t]
\SetAlgoLined
{\bf Input:} $\{y\sqb{k}\}_{k=0}^{K-1}$, $\tau,P$ in \eqref{gfs} and $M = |\mcal{M}|$ in \eqref{eq:SOCE}.\\
\KwResult{Samples $\RFD$ and bandlimited function, $\widetilde g\rob{t}$. }

\begin{enumerate}[label = $\arabic*)$,leftmargin=*,itemsep=0pt]
\item Compute $\bar{y} = \Delta y$ as in \eqref{eq:barvec}. 
\item Compute DFT or $\mdft{y} = \mat{V} \mat{\bar{y}}$ using \eqref{eq:dftvec}. 
\item Define $\sqb{\mat{z}}_n = \sqb{\mdft{y}}_n, \ n\in \mset{K-1} \setminus \eset{K-1}$.
\item Fold Estimation in the Fourier Domain Estimation.
\begin{enumerate}[label = $4 \alph*)$,leftmargin=*,itemsep=0pt]
  \item Using $\mat{z}$ from $3)$, define $\tpmat$ using \eqref{tmatrix}. 
  \item Find $\widetilde{\mat{p}}$ such that $\tpmat  \widetilde{\mat{p}} = \mat{0}$.
    \item Compute the roots $\widetilde \xi_m$ of the polynomial ${{\widetilde{\mathsf{P}}_M \rob{z}}}$ in analogy to \eqref{proots}. 
  \item Estimate $\widetilde{t}_m = - T \angle{\widetilde \xi}_m / \barwo$.
  \item Estimate $\widetilde c\sqb{m}$ using least-squares minimization, 
  \[\widetilde c\left[ m \right] \EQc{eq:SOCE} {\min _{c\left[ m \right]}}\sum\limits_n {{{\left| {z\sqb{n} - \sum\limits_{m\in\mcal{M}} {c\sqb{m}\widetilde{\xi}_m^n} } \right|}^2}}.\] 
\end{enumerate}
\item Estimate $\widetilde {\bar r}\sqb{k}$. This is done by plugging $\{\widetilde c\sqb{m}, \widetilde t_m\}_{m\in \mcal{M}}$ in \eqref{eq:SOCE}, yielding $\widehat{\widetilde{r}}\sqb{n}$ and then performing inverse DFT. 
\item Estimate $\widetilde r\sqb{k}$. 
\begin{enumerate}[label = $6 \alph*)$,leftmargin=*,itemsep=1pt]
  \item Zero-pad to convert $\widetilde {\bar r}\in \mset{K-1} \longrightarrow \widetilde {\bar r}\in\mset{K} $. 
  
	  That is, $\widetilde {\bar r}\sqb{k+1} = \widetilde {\bar r}\sqb{k}, k\in \mset{K-1}$ and $\widetilde {\bar r}\sqb{0} = 0$.
  \item Apply anti-difference, $\widetilde{r} \sqb{k} = \sum\limits_{m=0}^{k} \widetilde {\bar r}\sqb{m}, k\in\mset{K}$.

\end{enumerate}
\item Estimate $\RFD\sqb{k}=\widetilde { r}\sqb{k} + y\sqb{k}$ (up to an unknown const.).
\item Estimate $\widetilde g\rob{t}$ by applying sinc-interpolation to $\widetilde \gamma\sqb{k}$.  
\end{enumerate}
\caption{Fourier-Prony Recovery Algorithm.}
\label{alg:1}
\end{algorithm}

\subsubsection{Implementation Strategy}
\label{sec:IS}
When the conditions of the above theorem are met, the steps outlined in Algorithm~\ref{alg:1} recover the unfolded samples $\gamma\sqb{k}$ up to an unknown constant. The reconstruction is exact in the absence of noise. In the presence of noise and uncertainties arising from hardware implementation, the reconstruction procedure can be stabilized, for example, with the Matrix Pencil method \cite{Hua:1990}. Note however, that the Matrix Pencil method involves a tuning parameter (namely, the pencil parameter \cite{Hua:1990}), and finding the optimal choice for this parameter in a practical setting is typically not straight forward, especially as the number of spikes arising in our experiments exceeds the numbers commonly studied in previous works. As the Matrix Pencil method is not the main focus of this paper, for our experiments we use the choice that yields the best reconstruction.

\begin{figure*}[!t]
\centering
\includegraphics[width = 1\textwidth]{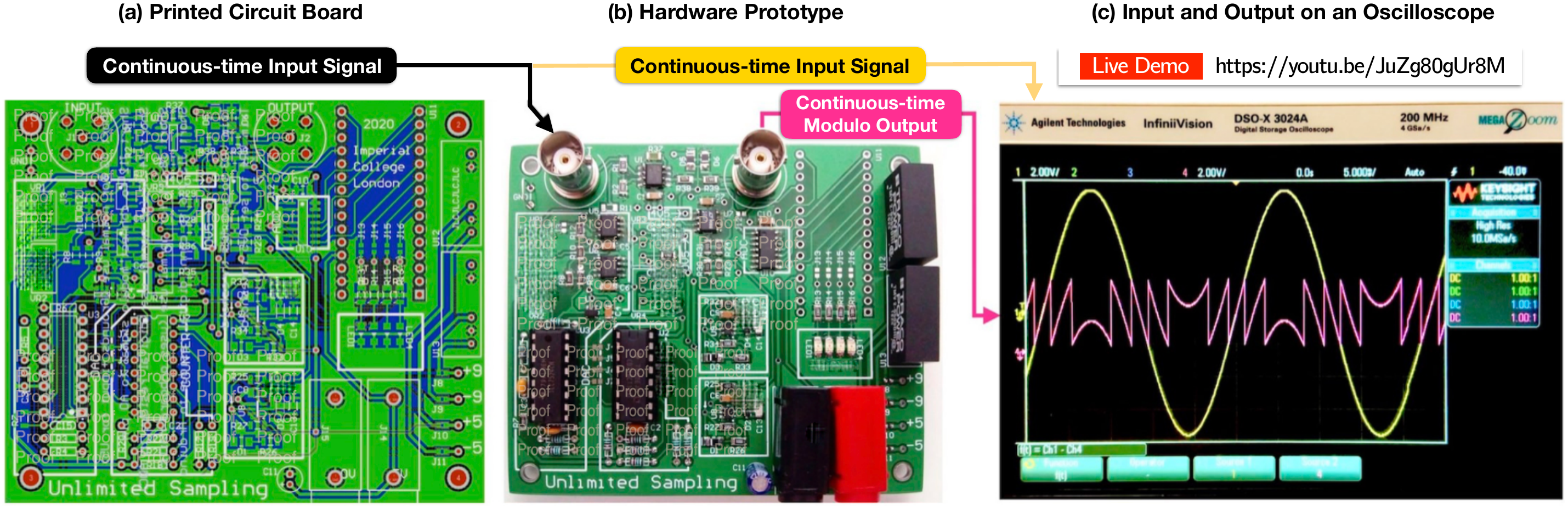}
\caption{Hardware prototype for unlimited sensing framework. Our initial design is capable of folding a signal that is as large as $24\lambda$ and as shown in Section \ref{sec:exp}, we have tested recovery of signals as large as $\vpp{48}$ and $\vpp{58}$, respectively. (a) Printed circuit board. (b) Electronic implementation that transforms a continuous-time input signal into a continuous-time modulo output. (c) Live screen shot of the oscilloscope plotting the output of a conventional ADC (yellow) and USF based ADC (pink) signal. When the input signal exceeds the oscilloscope's dynamic range, the oscilloscope's inbuilt ADC saturates. However, the output signal from our circuit continues to fold. A live {YouTube} demonstration of this hardware experiment is available at \href{https://youtu.be/JuZg80gUr8M}{\texttt{https://youtu.be/JuZg80gUr8M}}.}
\label{fig:HW}
\end{figure*}

\section{Hardware Experiments}
\label{sec:exp}
To investigate the validity of our approach we designed a hardware prototype---the \emph{US-ADC}---implementing the unlimited sampling pipeline. The electronic circuit, its hardware implementation together with an oscilloscope screenshot are shown in \fig{fig:HW}. To see an experiment in action, a live \texttt{YouTube} demonstration  has been made available at \href{https://youtu.be/JuZg80gUr8M}{\texttt{https://youtu.be/JuZg80gUr8M}}. In spirit of reproducible research, we plan to provide a ``\textsc{Do It Yourself}'' (DIY) hardware guide and release our algorithmic implementations in the near future. 


\medskip

\noindent{\bf Experimental Protocol.} {For each experiment, we simultaneously acquire the bandlimited signal $\BL{g}$ and its modulo samples $y\sqb{k}\in{\sqb{0.5,4.5}} + \mathsf{DC}$ (volts) where $\mathsf{DC}$ is an adjustable constant in the US-ADC. We use $\lambda \approx 2.01 \pm 3/20$ where $\pm 3/20$ is a manually adjustable design parameter. In experiments $1$--$4$, we use $\BL{g}$ in the range $\vpp{20}$ (volts, peak-to-peak) and $\mathsf{DC}$ is adjusted so that the modulo samples are aligned to the x-axis, taking both positive and negative values. In experiments $5$(a) and $5$(b), we use substantially HDR signals with a range of $\vpp{\approx48}$ and $\vpp{\approx58}$, respectively.} Although our hardware is equipped with its own ADC, to be able to obtain a ground truth, we simultaneously plot the input and output of the US-ADC on the $4$ channel DSO-X $3024$A oscilloscope. There on, we use the in-built $8$-bit sampler of the oscilloscope to sample the waveforms. Reflecting  the $8$-bit resolution, the samples are effectively quantized by an $8$-bit, uniform quantizer and hence, the measurements are corrupted by quantization noise following the model described in \cite{Bhandari:2020f}. For each experiment, we report the numerical values for,
\begin{enumerate}[leftmargin = *, label = $\bullet$]
  \item the experimental sampling rate $T$ and $\TFD$ in \eqref{eq:sampling} which is the sampling criterion required by Theorem~\ref{thm:FDUST}. Where needed, we also report $\TUS$ that guarantees that \eqref{eq:US} holds. 
    
  \item the dynamic range of the input signal $\rob{\DR{\gamma}}$ and the modulo samples $\rob{\DR{y}}$  where $\DR{z} = \max z\sqb{k} - \min z\sqb{k}$. The ratio $\DR{\gamma}/\DR{y}$ provides a measure of hight-dynamic-range recovery, independent of threshold $\lambda$.
\end{enumerate}
The bandwidth parameter $P$ in \eqref{gfs} is estimated from the ground truth. Since experimental data may not be exactly periodic, we use a slightly higher $P$ so that the complex exponentials are well isolated in Step $3)$ of Algorithm~\ref{alg:1} and one obtains an accurate reconstruction. We assume that $M = |\mcal{M}|$ is given and we use the matrix pencil method \cite{Hua:1990} in Step $4)$ of Algorithm~\ref{alg:1} (cf.~Section~\ref{sec:IS}). For performance evaluation, we compute the mean squared error (MSE) of the reconstruction defined in \eqref{eq:MSE}. We compare the following metrics. 
\begin{enumerate}[leftmargin = *, label = $\bullet$]
  \item $\mse{\RFD}$ the reconstruction MSE obtained by using our proposed approach in Algorithm~\ref{alg:1}. 
  \item $\mse{\RUS}$, the reconstruction MSE resulting from the unlimited sampling algorithm \cite{Bhandari:2020f} applied to modulo samples with a hardwired threshold value of $\lambda = 2.01$. 
  \item  $\mse{\RUSopt}$, the reconstruction MSE for the identical circuit design, but resulting from the unlimited sampling algorithm with a parameter $\lopt$ that is different from the hardwired modulo threshold $\lambda=2.01$, namely,
    \begin{equation}
    \label{eq:lopt}
    \lopt =   {\min\limits_{\widetilde{\lambda} }}\sum\limits_{k = 0}^{K - 1} {{{\left| {\gamma\sqb{k}  - \texttt{USRec}_{\widetilde{\lambda}}
    \rob{y} \sqb{k}} \right|}^2}} 
    \end{equation}
where $\RUS \sqb{k} = \texttt{USRec}_{\widetilde{\lambda}}\rob{y} \sqb{k}$ is the reconstruction due to unlimited sampling algorithm \cite{Bhandari:2020f}, for a given $\widetilde{\lambda}$. The rationale behind this choice is that in practice, where occasionally non-ideal folds may appear {(cf.~\fig{fig:demo1}(a) and the residue depicted in \fig{fig:demo1}(b))}, 
the hardwired threshold $\lambda$ will not always be the optimal parameter choice; in contrast \eqref{eq:lopt} will yield an optimal choice by design. 
\end{enumerate}
For all experiments, we use calibration to estimate the unknown offset arising from Step $6)$ of Algorithm \ref{alg:1}. The experimental parameters and results are summarized in Table~\ref{tab:1}. \smallskip

\begin{table*}[!t]
\centering
\caption{Summary of Experimental Parameters and Performance Evaluation}
\resizebox{\textwidth}{!}{%
\begin{threeparttable}
\begin{tabular}{@{}cccccccccccccc@{}}
\toprule
         \multicolumn{1}{c}{\multirow[t]{2}{*}{Exp.}} &
         \multicolumn{1}{c}{\multirow[t]{2}{*}{Fig.~No.}} &
         \multicolumn{1}{c}{\multirow[t]{2}{*}{$T$}}  &
         \multicolumn{1}{c}{\multirow[t]{2}{*}{$\TFD $}} &
         \multicolumn{1}{c}{\multirow[t]{2}{*}{$K$}} &
         \multicolumn{1}{c}{\multirow[t]{2}{*}{$\tau$}} &
         \multicolumn{1}{c}{\multirow[t]{2}{*}{$P$}} &
         \multicolumn{1}{c}{\multirow[t]{2}{*}{$\DR{\gamma}$}} &
         \multicolumn{1}{c}{\multirow[t]{2}{*}{$\DR{y}$}} &
         \multicolumn{1}{c}{\multirow[t]{2}{*}{$M$}} &
         \multicolumn{1}{c}{$\mse{\RFD}$} &
        $\mse{\RUS}$ &
        $\mse{\RUSopt}$ &
        \multicolumn{1}{c}{\multirow[t]{2}{*}{$\lopt$}} \\ [2pt]
        	\multicolumn{1}{c}{} &
	\multicolumn{1}{c}{} &
        	\multicolumn{1}{c}{$\rob{\mu\mathrm{s}}$} &
         \multicolumn{1}{c}{$\rob{\mu\mathrm{s}}$} &
         \multicolumn{1}{c}{} &
         \multicolumn{1}{c}{$\rob{\mathrm{ms}}$} &
         \multicolumn{1}{c}{$\left\lceil {\frac{\Omega }{{{\omega _0}}}} \right\rceil $} &
         \multicolumn{1}{c}{(V)} &
         \multicolumn{1}{c}{(V)} &
         \multicolumn{1}{c}{} &         
         \multicolumn{3}{c}{\cellcolor[HTML]{EFEFEF}(Reconstruction MSE)} &
          \multicolumn{1}{c}{}       \\ \midrule         
$1$ & \ref{fig:MTCS}, \ref{fig:ExpA} & $132.03$ & $517.86$ & $455$ & $60.07$ & $37$ & $19.87$ & $3.89$ & $20$ & $0.384\ep{-3}$ & $2.39\ep{-3}$ & $0.251\ep{-3}$ & $2.04$ \\
$2$ &  \ref{fig:demo1}(a), \ref{fig:ExpB} & $4$ & $21.652$ & $249$ & $0.996$ & $15$ & $7.71$ & $4.12$ & $7$ & $0.343\ep{-3}$ & $3.7$ & $18.7\ep{-3}$ & $1.48$ \\
$3$ &  \ref{fig:ExpC} & $1000$ & $2426.8$ & $199$ & $199$ & $14$ & $19.61$ & $3.99$ & $26$ & $5.92\ep{-3}$ & $3.38$ & $7.01$ & $0.83$ \\
$4$ &  \ref{fig:ExpD} & $81.25$ & $144.25$ & $245$ & $19.91$ & $20$ & $20.5$ & $3.97$ & $48$ & $9.38\ep{-3}$ & $1.416\ep{-2}$ & --- & --- \\
$5$(a) &  \ref{fig:ExpEa} & $108.49$ & $407.15$ & $1291$ & $140.1$ & $7$ & $\approx 48$ & $4.26$ & $161$ & --- & --- & --- & --- \\
$5$(b) & \ref{fig:ExpEb} & $70$ & $260.31$ & $357$ & $24.99$ & $3$ & $57.6$ & $5.44$ & $44$ & $0.3405$ & $1.0857$ & $0.4768$ & $2.07$ \\
\bottomrule
\end{tabular}%
\begin{tablenotes}
\scriptsize
\item $\bullet$ $T$ is the sampling rate of the ADC while $\TFD$ is the sampling rate criterion in Theorem~\ref{thm:FDUST}. $\bullet$ $\DR{\gamma}$ refers to the dynamic range of input signal (ground truth) and is computed using $\DR{\gamma} = \max \gamma\sqb{k} - \min \gamma\sqb{k}$. Similarly, the output signal dynamic range $\DR{y}$ refers to the modulo samples $y\sqb{k}$. $\bullet$ $\mse{\RFD}$ refers to the MSE due to reconstruction using the proposed Fourier domain approach. $\bullet$ $\mse{\RUS}$ refers to the MSE due to reconstruction using unlimited sampling approach with $\lambda = 2.01$. $\bullet$ $\mse{\RUSopt}$ refers to the MSE due to reconstruction using unlimited sampling approach with $\lopt$ obtained by \eqref{eq:lopt}.
\end{tablenotes}
\end{threeparttable}
}
\label{tab:1}
\end{table*}

\begin{figure*}[!t]
\centering
\includegraphics[width = 1\textwidth]{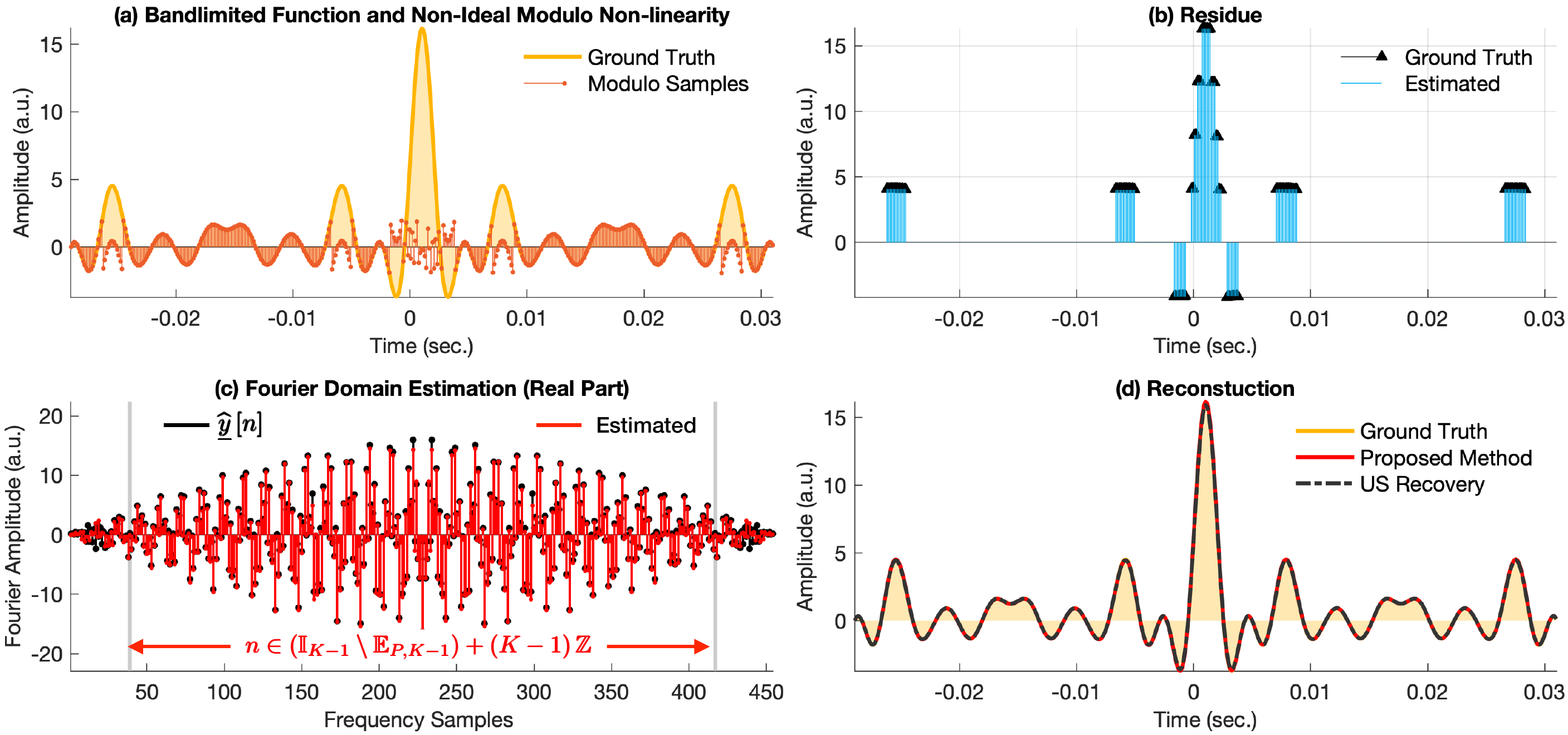}
\caption{Experiment 1: Backwards compatibility with unlimited sampling. (a) Ground truth signal and modulo samples. {The $\mathsf{DC}$ has been adjusted so that the modulo samples are aligned with the x-axis.} (b) Ground truth residue and recovered residue. (c) Fourier domain estimation of $\dft{r}{n}$; as desired, it approximately agrees with $\dft{y}{n}$ for $n\in \mset{K-1} \setminus \eset{K-1}$. (d) Reconstruction using proposed approach agrees with the unlimited sampling method $\rob{\lambda = 2.01}$. }
\label{fig:ExpA}
\end{figure*}

\begin{figure}[!t]
\centering
\includegraphics[width = 0.65\columnwidth]{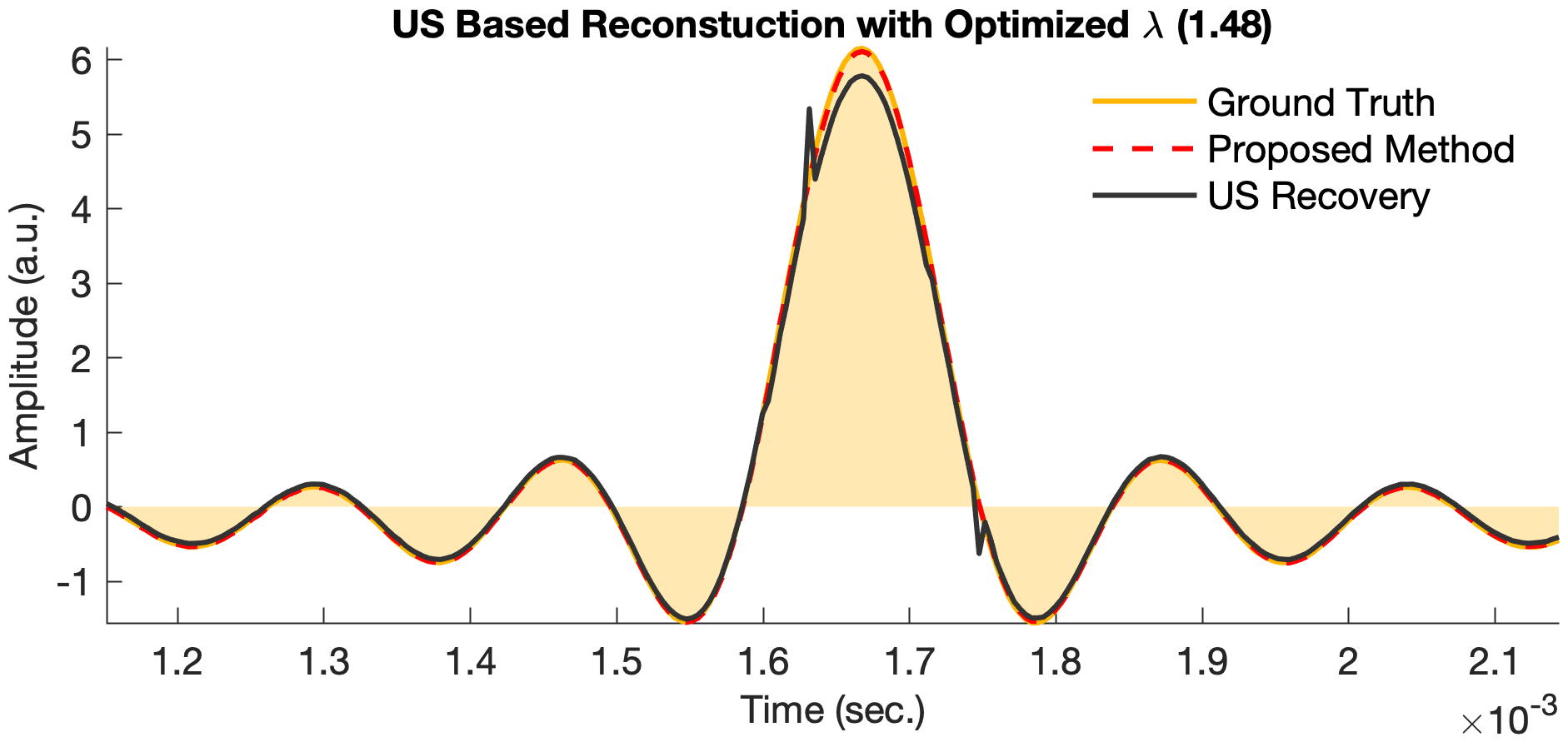}
\caption{Experiment 2: Optimizing for $\lambda$ enhances the performance of the unlimited sampling method. In this case, $\lopt = 1.48$ and using this value, the reconstruction falls steeply from $\mse{\RUS} = 3.7$ to $\mse{\RUSopt} = 18.7\times10^{-3}$. For comparison with reconstruction using the hardwired value $\lambda = 2.01$, see \fig{fig:demo1}(d). Having obtained $\lopt$, we note that the breaking points for unlimited sampling method are the sparse locations in \fig{fig:demo1}(b) where the amplitudes deviate the largest from the grid $2\lambda\Z$. These locations can be accurately estimated by the Fourier approach of this paper.}
\label{fig:ExpB}
 \end{figure}

\medskip

\noindent {\bf Experiment 1: Backwards Compatibility with the Unlimited Sampling Approach.} The goal of this experiment is to show that in its finite dimensional setting, our recovery approach is backwards compatible with the unlimited sampling approach. For this experiment, we use modulo samples of a randomly generated trigonometric polynomial shown in \fig{fig:MTCS}. The experimental data and the corresponding recovery results are shown in \fig{fig:ExpA}. The signal with $\tau = 60.07 \times 10^{-3}$ s is sampled with $T = 132.03$ $\mu$s yielding $K = 455$ samples. The sampling rate required for recovery is $\TFD = 517.86$ $\mu$s. We use $P = 37$. The input signal dynamic range is $\DR{\gamma} = \vpp{19.866}$ and the same for folded samples is $\DR{y} = \vpp{3.8852}$. The ratio $\DR{\gamma} /\DR{y}  = 5.1132$ shows that a signal as large as $\approx 10$ times the US-ADC threshold can be recovered. The experimental data approximately satisfies the unlimited sampling hypothesis, namely the condition in \eqref{eq:US} and this results in a reconstruction MSE of $\mse{\RUS}= 2.385\times 10^{-3}$. Quantization and system noise lead to inaccuracies specially around $t=0$ where folding is concentrated. The proposed approach, with $M = 20$ folds, results in $\mse{\RFD}= 3.838\times 10^{-4}$ which is a factor $10$ improvement in the MSE. This performance is comparable to $\mse{\RUSopt}= 2.385\times 10^{-4}$ which is obtained by optimizing $\lambda$ using \eqref{eq:lopt}, which turns out to be $\lopt = 2.04$.

\medskip
\noindent {\bf Experiment 2: Moderate Number of Non-ideal Folds; Towards of a Hybrid Reconstruction Approach.} 
We generate a Dirichlet kernel (periodized sinc function). 
The experimental data is shown in \fig{fig:demo1} and in this case, $\tau = 9.96 \times 10^{-4}$ s. 
The signal is sampled with $T = 4$ $\mu$s yielding $K = 249$ samples. 
The sampling rate predicted by \eqref{eq:sampling} is $\TFD = 21.652$ $\mu$s.
We estimate $P = 15$. The input signal dynamic range is $\vpp{7.7098}$ and the same for folded samples is $\vpp{4.1171}$. 
The values are specifically chosen to evaluate the performance of the algorithm with a smaller number of non-ideal folds $\rob{M =7}$. 
Despite the non-idealities, our Algorithm \ref{alg:1} is able to accurately reconstruct the signal resulting in a reconstruction MSE of $\mse{\RFD} = 3.43 \times 10^{-4}$. 
In contrast, even though the sampling rate in this experiment satisfies the condition in \eqref{eq:US} (numerically), the reconstruction breaks down when using unlimited sampling algorithm \cite{Bhandari:2020f} and ${\mse{\RUS} = 3.7}$. 
This performance can be enhanced by optimizing $\lambda$, in which case we observe $\mse{\RUSopt} = 1.87\times10^{-2}$ which greatly improves up on the unlimited sampling algorithm with the hardwired parameter but remains two orders of magnitude worse than Algorithm 1. This worse performance is mainly due to a few ``breaking points'' between which one encounters a temporary offset. These breaking points correspond to the few locations in \fig{fig:demo1}(b) where the amplitudes deviate the most from the grid $2\lambda\Z$. However, these locations can be exactly estimated using our Fourier domain approach that is agnostic to $\lambda$. 
This shows promise for a hybrid reconstruction approach where the unlimited sampling method is used to resolve most of the folds followed by a Fourier domain approach to resolve the remaining breaking points. 

 \begin{figure*}[!t]
\centering
\includegraphics[width = 1\textwidth]{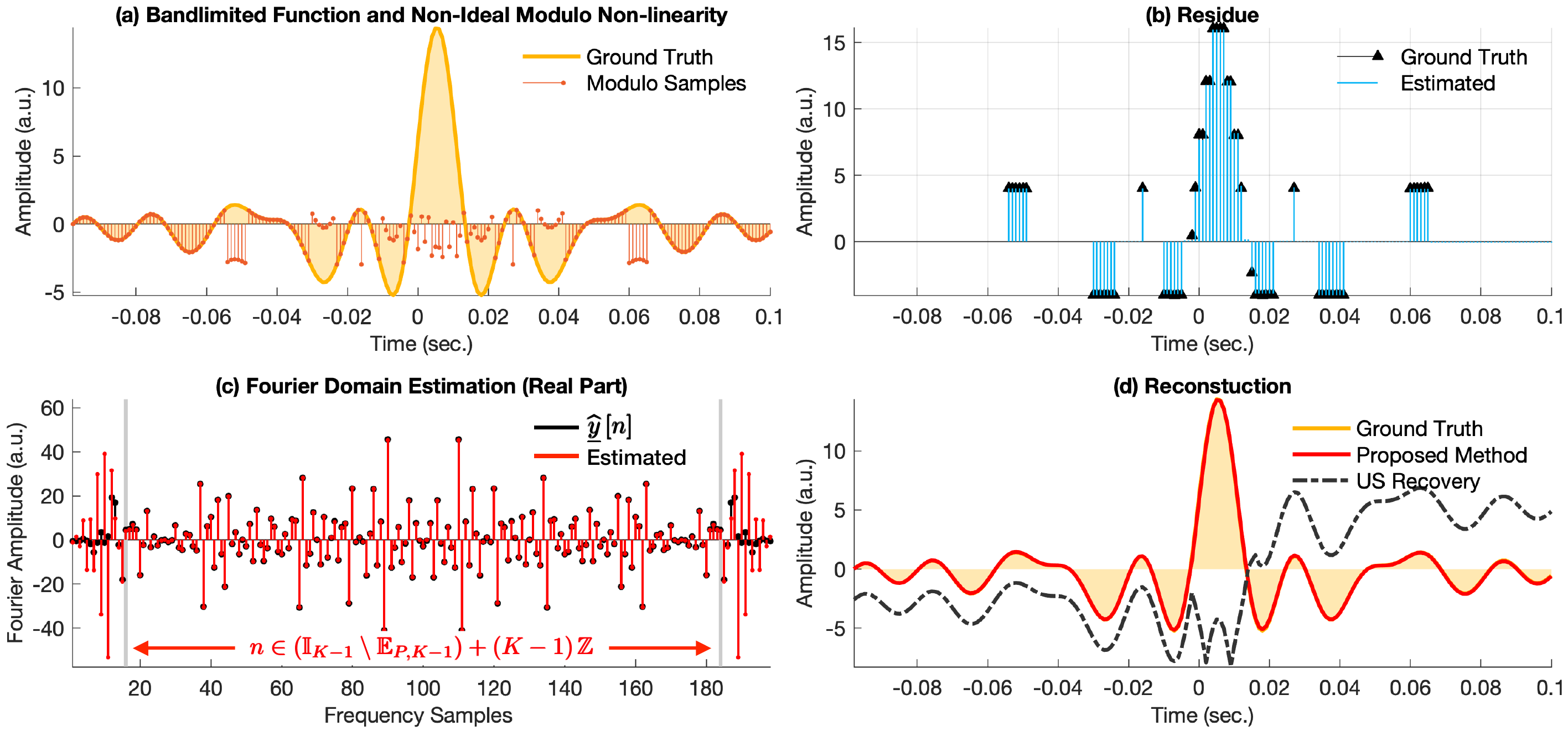}
\caption{Experiment 3: Recovery at lower sampling rates than what is prescribed by the Unlimited Sampling Theorem in Theorem~\ref{thm:UST} . (a) Ground truth and non-ideal modulo samples. (b) Ground truth and recovered residue. (c) Fourier domain estimation of $\dft{r}{n}$, again showing approximate agreement with $\dft{y}{n}$ for $n\in \mset{K-1} \setminus \eset{K-1}$. (d) Reconstruction using the proposed approach and the unlimited sampling method $\rob{\lambda = 2.01}$. }
\label{fig:ExpC}
\end{figure*}

\begin{figure*}[!t]
\centering
\includegraphics[width = 1\textwidth]{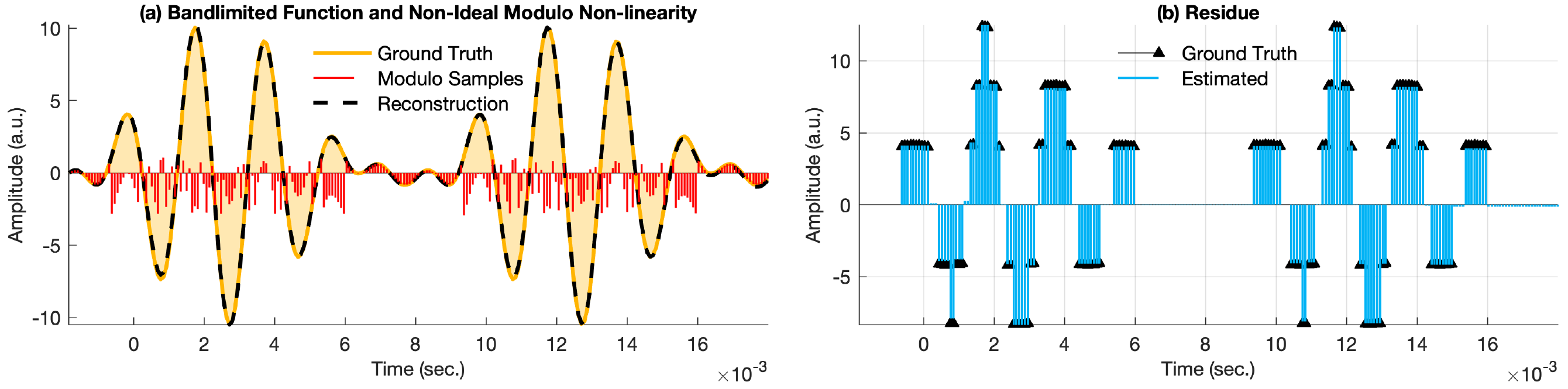}
\caption{Experiment 4: Burst signal with higher number of folds ($M = 48)$. (a) Ground truth signal and non-ideal modulo samples. (b) Ground truth residue and recovered residue. }
\label{fig:ExpD}
 \end{figure*}

\medskip
\noindent {\bf Experiment 3: Recovery where Unlimited Sampling Requires Higher Order Differences.} 
Given that the method proposed in this paper relies only on the first order differences of the samples, one could expect that sampling rate $T$ needs to be of size comparable to what is required for the unlimited sampling method with first order differences, $N= 1$. This experiment however confirms our theoretical finding that this is not the case. Despite a sampling rate for which the unlimited sampling method requires higher order differences, the Fourier based approach still yields accurate recovery -- as predicted by Theorem~\ref{thm:FDUST}. The experimental parameters are listed in Table~\ref{tab:1} and the reconstruction is shown in \fig{fig:ExpC}. Indeed, in this experiment, the choice of sampling rate $T = 1$ ms violates the numerical condition $|\Delta x| \leq \lambda$ which is necessary to allow for the choice $N=1$, but satisfies the condition of  Theorem~\ref{thm:FDUST}.

\medskip
\noindent {\bf Experiment 4: Recovery of Burst Signals with Clustered Folds.} In this example, we consider a ``burst" signal which introduces clustered folds. Such signals typically arise in digital and radio communications where \emph{amplitude modulation} is used for transmitting messages. From the experimental parameters in Table~\ref{tab:1}, we note that in comparison to previous setups,  this case results in a relatively high number of folds, $M = 48$, which are also clustered. This is a challenge in the super-resolution step of our algorithm, as such methods work best for few, well-separated spikes (corresponding to folds in our measurements). Nevertheless, with reasonable oversampling, $\rob{\TFD/T} \approx 1.78$, our recovery approach is able to reconstruct the signal accurately with $\mse{\RFD} = 9.37\times 10^{-3}$. The reconstruction is shown in \fig{fig:ExpD}(a) and the recovered residue is shown in \fig{fig:ExpD}(b).

 \begin{figure}[!t]
\centering
\includegraphics[width = 0.65\textwidth]{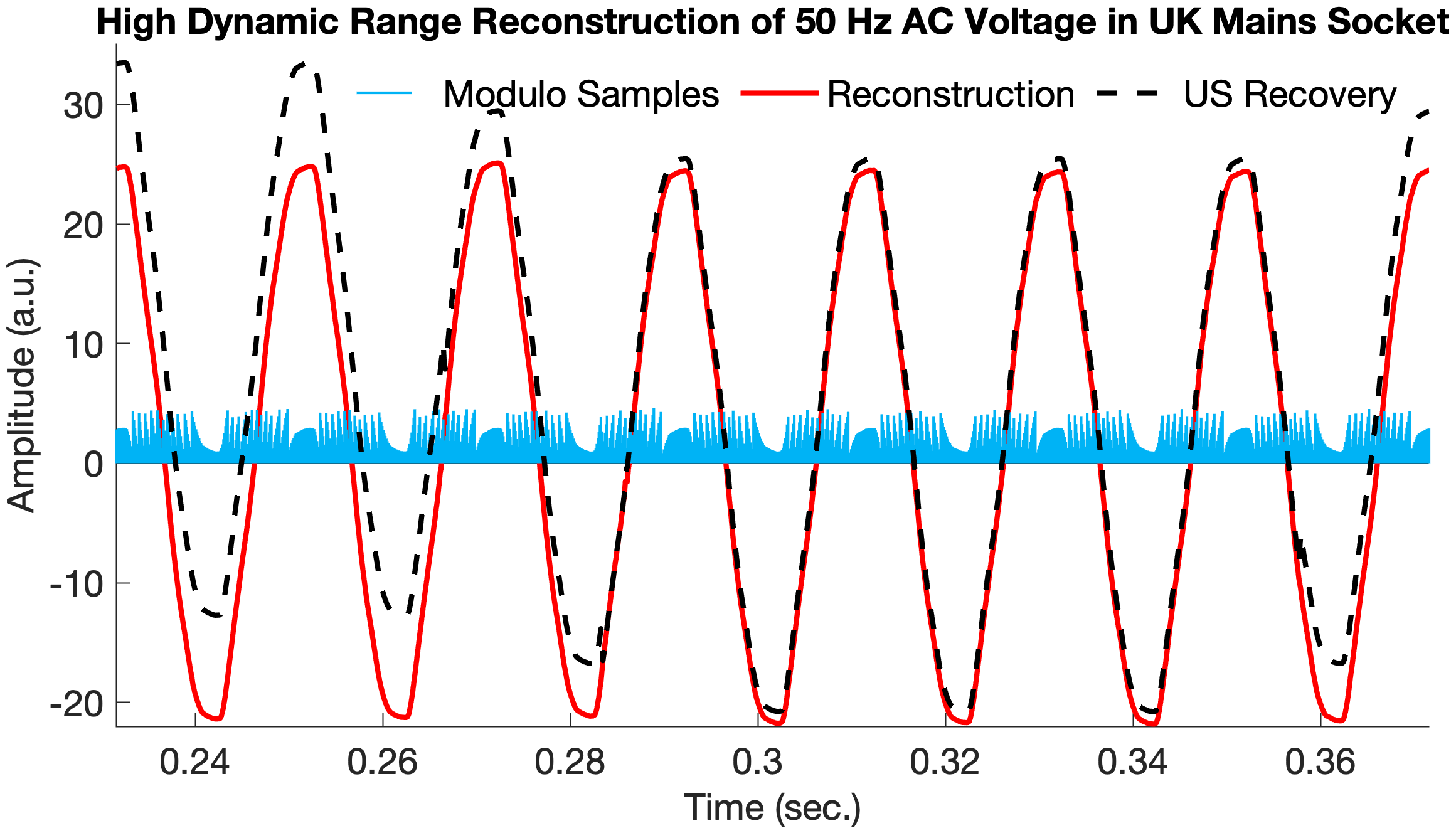}
\caption{Experiment 5a: High Dynamic Range voltage reconstruction of $\approx \vpp{48}$ UK mains alternating current ($50$ Hz). }
\label{fig:ExpEa}
\end{figure}

\medskip
\noindent {\bf Experiment 5: High Dynamic Range Reconstruction.} 

\noindent In the two experiments that follow, our goal is to push our hardware and the algorithm to its limits so that signals with amplitudes as large as $24\lambda$  and folds as large as $M = 161$, can be reconstructed. Such HDR inputs are likely to amplify deviations and non-idealities in the US-ADC prototype and hence, the experiments serve as an edge case test for our recovery algorithm. 

\medskip
\noindent  (a) {\bf Uncalibrated Example $\rob{\approx \vpp{48}}$.} In this case, we use our hardware's internal ADC to sample the waveform and hence, we do not have access to the ground truth. That said, the signal of interest is the alternating current drawn from the UK mains power socket with a frequency $\approx50$ Hz. Out of $K = 1291$ samples are sampled at $T=108.488$ $\mu\mathrm{s}$. Due to the high dynamic range, we estimate $M = 161$ folding instants. Our reconstruction approach gives a reasonable reconstruction and to check this, we observe the Fourier spectrum of the reconstructed signal which shows a spike at $50.018$ Hz. The modulo samples and the corresponding reconstruction is shown in \fig{fig:ExpEa}. We also show that recovery using unlimited sampling method fails due to non-idealities. {We find it remarkable that our approach yields and accurate reconstruction despite the fairly large number of folds $(M = 161)$.  Our explanation for this performance is the interplay between accurate modeling, exact knowledge of $M$, and oversampling that avoids algorithmic challenges.}

\medskip

\noindent (b) {\bf  Calibrated Example $\rob{\approx \vpp{58}}$.} To establish that our recovery approach can indeed handle HDR signals, we repeat the experiment with access to the ground truth as we use the oscilloscope's built-in ADC. 
The corresponding waveforms and reconstruction are plotted in \fig{fig:ExpEb}(a). 
Due to the HDR swing of the input signal, non-ideal jumps are observed in the measurements and have been annotated in \fig{fig:ExpEb}(b). 
The non-ideal jumps result in sub-optimal reconstruction when using the unlimited sampling method but the performance can be enhanced by optimizing $\lambda$.
The results tabulated in Table~\ref{tab:1} yet again show the effectivity of our approach.

\begin{figure}[!t]
\centering
\includegraphics[width = 0.65\columnwidth]{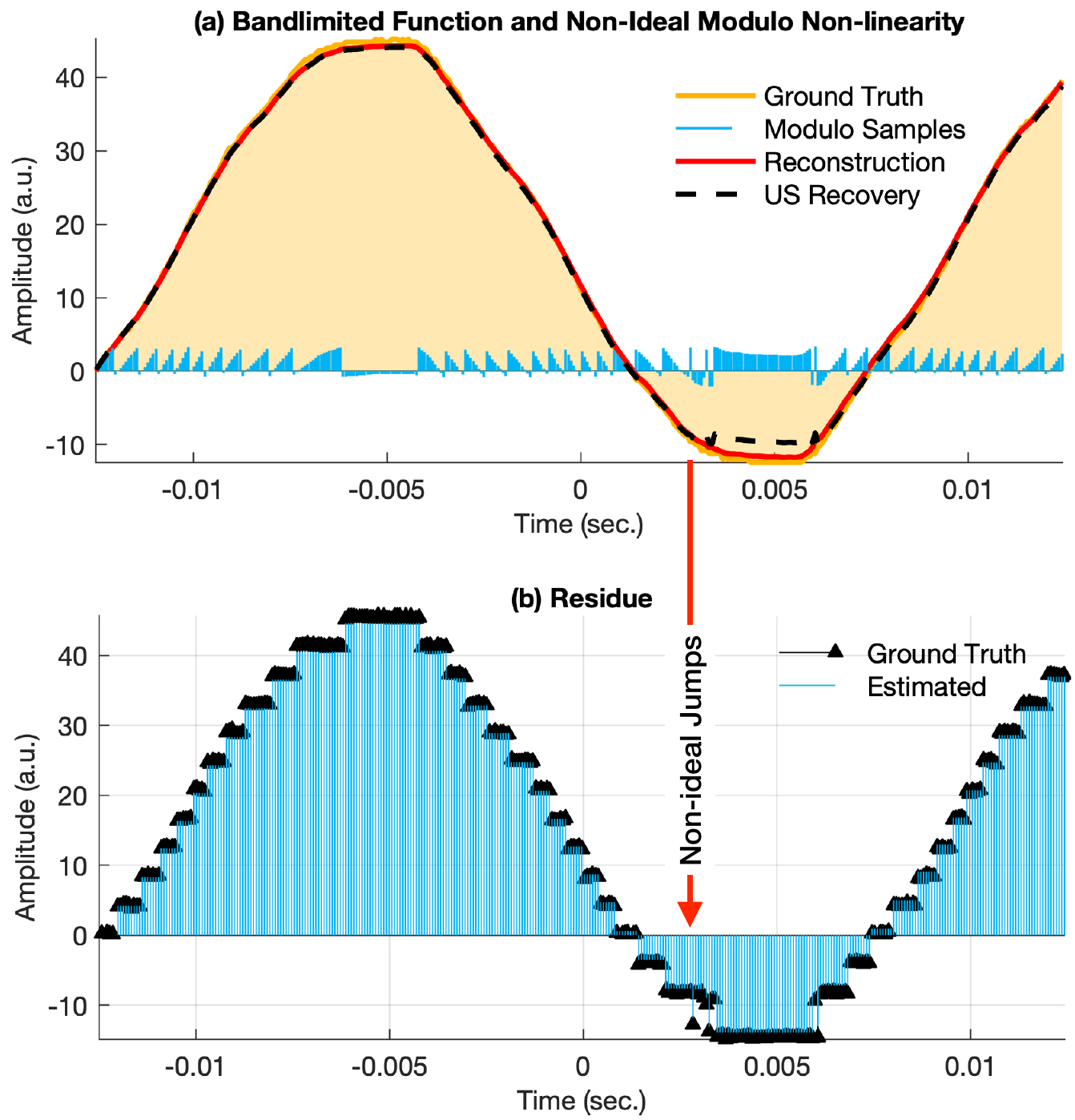}
\caption{Experiment 5b: HDR voltage reconstruction of $\approx \vpp{58}$ signal. }
\label{fig:ExpEb}
 \end{figure}

\section{Conclusions and Take-Home Message}
In our previous works on unlimited sampling \cite{Bhandari:2017b,Bhandari:2020g,Bhandari:2020f}, we studied new acquisition protocols and recovery algorithms for high-dynamic-range sensing based on modulo measurements. 	In this paper, we revisited this problem and proposed a novel solution approach, which, in contrast to the first works, is robust to non-idealities, as we observed them in experiments with a hardware prototype that we developed.

Our new algorithm is designed for finite-dimensional, folded signals and is agnostic to the sensing threshold $\lambda$. The main insight behind our approach is that the folds introduced by the modulo non-linearities can be isolated in the Fourier domain, which gives rise to a frequency estimation problem. For recovery, we rely on high resolution spectral estimation methods. This allows us to deal with arbitrarily close folding instants. At the cross-roads of theory and practice, our work raises interesting questions for future research. 
\begin{enumerate}[leftmargin = *,label = $\bullet$]
  \item We currently assume that the number of folds is known. We find it very interesting to explore whether this number can be bounded in terms of function parameters such as the amplitude. Alternatively, developing a robust criterion for estimating the same from data would benefit the recovery procedure. 
  
  \item Although we have presented empirical results based on experiments with a hardware prototype, our analysis does not yet consider the case of noise for the Fourier domain approach.  This remains an interesting pursuit to complement our guarantees.
  
  \item At the core of the recovery procedure designed in this paper is a spectral estimation problem \cite{Hua:1990}. We expect that future advances for this problem will also have interesting implications for the problem of reconstruction from modulo measurements. In particular, the limitations of current approaches for this problem in terms of the number of spikes that can be recovered will also directly translate into limitations of the approach presented in this paper. Also viewing the problem from the perspective of \textit{super-resolution} \cite{Donoho:1992} may yield additional insights and solution strategies.
  
  \item As an alternative way to overcome these limitations, in future work we aim to investigate hybrid methods that use the Fourier domain approach only for spikes that correspond to non-idealities and combine it with the original unlimited sampling method for the other spikes. 
\end{enumerate}

\ifCLASSOPTIONcaptionsoff
\newpage
\fi

\bibliographystyle{IEEEtran_url}

\end{document}